\documentclass[journal]{IEEEtran}
\usepackage{t1enc}
\usepackage[utf8]{inputenc}
\usepackage[english]{babel}
\usepackage{array}
\usepackage{hhline}
\usepackage{amsmath}
\usepackage{ulem}
\usepackage{amssymb}
\usepackage{ntheorem-hyper}
\usepackage{enumerate}
\usepackage{graphicx}
\usepackage{graphics}
\usepackage{psfrag}
\newcommand{\ud}{\mathrm{d}}

\newcommand{\sumfrac}[2]{\genfrac{}{}{0pt}{}{#1}{#2}}

\newcommand{\Vx}{\underline{x}}

\newcommand{\Va}{\underline{a}}
\newcommand{\Vb}{\underline{b}}
\newcommand{\Vr}{\underline{r}}

\newcommand{\Vy}{\underline{y}}
\newcommand{\Vz}{\underline{z}}

\newcommand{\Vl}{\underline{l}}

\newcommand{\irott}{\mathcal}
 
\renewcommand{\Pr}{Pr}
 
 \DeclareMathOperator{\I}{I}
\DeclareMathOperator{\Hh}{H}

\DeclareMathOperator{\argmax}{argmax}

\DeclareMathOperator{\argmin}{argmin}

\newcommand{\integ}{\int\limits} 
\theoremstyle{change}
\newtheorem{Def}{Definition}
\theoremstyle{break}
\newtheorem{remark}{Remark}
\theoremstyle{plain}
\newtheorem{Tet}{Theorem}
\theoremstyle{break}
\newtheorem{Lem}{Lemma}

\title{Blind decoding of Linear Gaussian channels with ISI, capacity, error exponent, universality}
\author{Farkas Lóránt\thanks{Thanks for Imre Csiszár for his numerous correction in this work.}}

\begin{document}

\maketitle

\begin{abstract}
 A new straightforward universal blind detection algorithm for linear Gaussian channel with ISI is given. A new error exponent is derived, which is better than Gallager's random coding error exponent. 
\end{abstract}

\section{Introduction}

In this paper, the discrete Gaussian channel with intersymbol interference (ISI)
\begin{align}
  y_i= \sum_{j=0}^l x_{i-j}h_j+z_i  \label{modell}
 \end{align}
will be considered, where the vector $h=(h_0,h_1,\dots,h_l)$ represents the ISI, and $\{z_i\}$ is white Gaussian noise with variance $\sigma^2$

A similar continuous time model has been studied in Gallager \cite{Gallager}. He showed that it could be reduced to the form
\begin{align}
y_n =v_n x_n + w_n  \label{eleje}
\end{align}
where the $v_n$ are eigenvalues of the correlation operator. The same is true also for the discrete model (\ref{modell}), but the reduction requires knowledge of the covariance matrix $R(i,j)=\sum_{k=0}^l h(k-i)h(k-j)$ whose eigenvectors should be used as new basis vectors.

Here however such knowledge will not be assumed, our goal is to study universal coding for the class of ISI channels of form (\ref{modell}).

As a motivation , note that the alternate method of first identifying the channel by transmitting a known ``training sequences'' has some drawbacks.
Because the length of the training sequence is limited, the estimation of the channel can be imprecise, and the data sequence is thus decoded according to an incorrect likelihood function. This results in an increase in error rates \cite{incorr_error_exp1}, \cite{incorr_error_exp2} and in a decrease in capacity \cite{incorr_error_cap1}. As the training sequence contains no valuable information, the longer it is the less information bits can be carried.

One can think this problem could be solved by choosing the training sequence sufficiently large to ensure precise channel estimation, and then choose the data block sufficiently long, but this solution seldom works due to the delay constraint, and to the slow change in time of the channel.

So we will give a straightforward method of coding and decoding, without any Channel Side Information (CSI) 
To achieve this, we generalise the result of Csiszár and Körner \cite{Csiszar} to Gaussian channel with white noise and ISI, using an analogue of the Maximal Mutual Information (MMI) decoder in \cite{Csiszar}. 

We will show that our new method is not only universal, not depending on the channel, but its error exponent is better in many cases than Gallager's \cite{Gallager} lower bound for the case when complete CSI is available to the receiver.
Previously, Gallager's error exponent has been improved for some channels, using an MMI decoder, such as for discrete memoryless multiple-access channels \cite{Pokorny/Wallmeier}.

We don't use Generalised Maximum likelihood decoding \cite{Ziv}, but a generalised version of MMI decoding. This is done by firstly approximating the channel parameters by maximum likelihood estimation, and then adopt the message whose mutual information with the estimated parameters is maximal. By using an extension of the powerful method of types, we can simply derive the capacity region, and random coding error exponent. At the end, we have a more general result, namely: We show how the method of types can be extended to a continuous, non memoryless environment.

The structure of this correspondence is as follows. In Section \ref{Definitions} we generalise typical sequences to ISI channel environment. The main goal of section \ref{Theorem} is to give a new method of blind detection. In section \ref{numericalr} we show by numerical results, that for some parameters the new error exponent is better than Gallager's random coding error exponent. In Section \ref{Discussion} we discuss the result, and give a general formula to the channels with fading.
\section{Definitions} \label{Definitions}
Let $\gamma_n$ be a sequence of positive numbers with limit 0. The sequence $\Vx \in \mathbb R^n$ is $\gamma_n$-typical to an n-dimensional continuous distribution $P$, denoted by $\Vx \in \irott{T}_P$, if
\begin{equation}
 \frac{\left| - \log(p(\Vx)) - H(P)  \right|}{n} < \gamma_n \label{tipikussag}
\end{equation}
where $p(\Vx)$ denotes the density function, and $H(P)$ the differential entropy of $P$.

Similarly sequences $\Vx \in \mathbb R^n$,$\Vy \in \mathbb R^{n+l}$ are jointly $\gamma_n$-typical to $2n+l$ dimensional joint random distribution $P_{X,Y}$ denoted by $(\Vx,\Vy) \in \irott{T}_{P_{XY}}$, if
\[ \left| - \log(p_{X,Y}(\Vy,\Vx)) - H(p_{X,Y})  \right| < n \gamma_n \]

In the same way, a sequence $\Vy$ is $\gamma_n$-typical to the conditional distribution $P_{Y|X}$, given that $X=\Vx$, denoted by $\Vy \in \irott{T}_{P_{Y|X}}(\Vx)$ if
\[ \left| - \log(p_{Y|X}(\Vy|\Vx))- H(Y|X)  \right| < n \gamma_n \]

For simplifying the proof, in the following, $P_{X}=P$ is always the $n$ dimensional i.i.d. standard normal distribution, the optimal input distribution of a Gaussian channel with power constrain 1. 
The conditional distribution will be chosen as $P^{(n,\sigma)}_{Y|X}$ with density \[\frac{1}{\sigma^n(2\pi)^{(n+l_n)/2}}exp(\frac{-|\Vy-h*\Vx|^2}{2\sigma^2})\]
where $h=(h_0,h_1,h_2,\dots,h_l)$ and $(h*x)_
i=\sum_{j=0}^l x_{i-j}h_j$ where $x_k=0$ is understood for $k<0$. 
So in this case 
\[H_{h,\sigma}(Y|X)=H(h*X+Z|X)=H(Z)=n[\frac{\ln(2\pi e)}{2}+\ln(\sigma)]\]

The limit of the entropy of $Y=X*h+Z$ as $n\rightarrow \infty$, is
\[\lim_{n \rightarrow \infty}\frac{\Hh_{h,\sigma}(Y)}{n}= \frac{1}{2} \ln(2\pi e)+  \frac{1}{2\pi}\int_0^{2\pi} \ln\left( \sigma+f(\xi)\right) d\xi\]
where $f(\lambda)=\sum_{k=-\infty}^{\infty}(\sum_{j=0}^{l-|k|}h_{j}h_{j+|k|})e^{ik\lambda}$ see \cite{Toeplitz} (here  $R_{m,n}=r(m-n)=r(k)=\sum_{j=0}^{l-|k|}h_{j}h_{j+|k|}$ is the correlation matrix). So the limit of the average mutual information per symbol, that is
\[\lim_{n \rightarrow \infty} \I^n(h,\sigma)=\lim_{n\rightarrow \infty} \frac{\Hh_{h,\sigma}(Y)-\Hh_{h,\sigma}(Y|X)}{n} \]
is equal to
\[I(h,\sigma) \circeq \frac{1}{2\pi}\int_0^{2\pi} \ln\left(1+ \frac{f(\xi)}{\sigma}\right) d\xi \]
moreover the sequence $I^n(h,\sigma)$ is non-increasing (see \cite{Toeplitz}).

We will consider a finite set of channels that grows subexponentially with $n$, and in the limit dense in the set of all ISI channels.
To this end, define the set of approximating ISI, as
\begin{align*} \irott{H}_n = \{h \in \mathbb{R}^{l_n}:h_i=k_i\gamma_n, |h_i| < P_n,\, k_i\in \mathbb{Z}, \\
\forall i \in \{1,2,\dots,l_n\} \}\end{align*}
where $l_n$ is the length of the ISI, $P_n$ is the power constraint per symbol,  and $\gamma_n$ is the ``coarse graining'', intuitively the precision of detection.
Similarly we define the set of approximating variances as
\begin{align*}
 \irott{V}_n = \{\sigma \in \mathbb{R}:\sigma=k\gamma_n, 1/2 < |\sigma| < P_n,\, k\in \mathbb{Z}^+ \\ \forall i \in \{1,2,\dots,l_n\}\}
\end{align*}
These two sets form the approximating set of parameters, denoted by $\irott S_n = \irott H_n \times \irott V_n$.

Below we set $l_n=[\log_2(n)]$, $P_n=n^{\frac{1}{16}}$, $\gamma_n=n^{-\frac{1}{4}}$

\begin{Def}\label{ISItype}
 The ISI type of a pair $(\Vx,\Vy)\in \mathbb R^n \times \mathbb R^{n+l_n}$ is the pair $(h_n,\sigma_n)\in \irott S_n$ defined by
\begin{align*}
h(i)=\argmin\limits_{h(i)\in \irott{H}_n} \|\Vy - h(i)*\Vx_i\|\\
\sigma(i)^2=\argmin\limits_{\sigma(i)\in \irott{V}_n}|\sigma(i)^2- \min_{h(i)\in \irott{H}_n}\frac{\|\Vy-h(i)*\Vx_i\|^2}{n}|
\end{align*}
\end{Def}
Note that this type
concept does not apply to separate input or output sequences, only to
pairs $(\Vx,\Vy)$.

\section{Lemmas, Theorem} \label{Theorem}
We summarise the result of this section: The first Lemma shows that the above definition of ISI type is consistent, in the sense that $\Vy$ is conditionally $P^{h,\sigma}_{Y|X}$ typical given $\Vx$, at least in the case when $\|\Vy-h*\Vx\|^2$ is not too large. Lemma \ref{Lem0} gives the properties which we need for our method, and proves that almost all randomly generated sequences has these properties. Lemma \ref{Lem1} gives an upper bound to the set of output signals, which are ``problematic'' thus typical to two codewords, namely they can be result of two different codeword with different channel. Lemma \ref{lem3} shows that if the channel parameters estimated via maximum likelihood (ML), the codewords and the noise cannot be very correlated. Lemma \ref{lem4} gives the formula of the probability of the event that an output sequence is typical with another codeword with respect to another channel. All Lemmas are used in Theorem \ref{mainthm}, which gives the main result, and defines the detection method strictly.

\begin{Lem}
  When 
\[|\frac{\|\Vy-h(i)*\Vx_i\|^2}{n}-\sigma(i)^2|<\gamma_n\]
so the detected variances is in the interior of the set of approximating variances, then 
\[\Vy \in \irott{T}_{P_{Y|X}^{h(i),\sigma(i)}} (\Vx_i)\]
\end{Lem}
\begin{proof}
Indeed, if 
\[|\frac{\|\Vy-h(i)*\Vx_i\|^2}{n}-\sigma(i)^2|<\gamma_n\]
then
\[\left| \frac{\|\Vy-h(i)*\Vx_i\|^2}{2\sigma(i)^2} - \frac{n}{2}\right|<n\gamma_n.\]
With
\begin{align*}
-\log(P_{Y_|X}^{h(i),\sigma(i)}(\Vy|\Vx_i))=\frac{n}{2}\log(2\pi\sigma(i)^2)+\frac{\|\Vy-h(i)*\Vx_i\|^2}{2\sigma(i)^2}\\
\end{align*}
we get
\begin{align*}
|-\log(P_{Y_|X}^{h(i),\sigma(i)}(\Vy|\Vx_i))-H_{P_{Y|X}^{h(i),\sigma(i)}}(Y|X)|=\\
\left| \frac{\|\Vy-h(i)*\Vx_i\|^2}{2\sigma(i)^2} - \frac{n}{2}\right|  
\end{align*}

and by the definition $\Vy \in \irott{T}_{P_{Y|x}^{h(i),\sigma(i)}}(\Vx_i)$ if
\[|-\log(P_{Y_|X}^{h(i),\sigma(i)}(\Vy|\Vx_i))-H_{P_{Y|X}^{h(i),\sigma(i)}}(Y|X)|< n\gamma_n \]
\end{proof}
\begin{Lem} \label{Lem0}
For arbitrarily small $\delta>0$, if $n$ is large enough, there exists a set $\irott{A} \subset \irott T_{P}^n$, with $P^n(\irott A)>1-\delta$, where $P$ is the $n$-dimensional standard normal distribution, such that for all $\Vx \in \irott{A}$
,$k,l \in \{0,1,\dots,l_n\}$ $k\neq l$
\begin{align}
& \left|\frac{\sum_{j=0}^n x_{j-k}x_{j-l}}{n}\right|< \gamma_n \label{lem0-0}\\
& \left|\frac{\sum_{j=0}^n x_{j-k}x_{j-k}}{n}-1\right|< \gamma_n\label{lem0-1}\\
& \left|-\frac{1}{n}\ln p(\Vx)-1/n \Hh(P)\right|< \gamma_n \label{lem0-2}
\end{align}
\end{Lem}
\begin{proof}
Take  $n$ i.i.d, standard Gaussian random variables $X_1,X_2,\dots,X_n$. Fix $k,l$ $1<k,l<n$. By Chebishev's inequality,
\[ \Pr\left\{\left|\frac{\sum_{i=1}^n X_{i-k}X_{i-l}}{n}\right|>\xi \sqrt{\frac{1}{n}}\right\}<\frac{1}{\xi ^2}\]
From this, with $\xi=\gamma_n n^{\frac{1}{2}}$ 
\[\Pr\left\{\left|\frac{\sum_{i=1}^n X_{i-k}X_{i-l}}{n}\right|>\gamma_n \right\}<\frac{1}{(\gamma_n n^{\frac{1}{2}})^2}=\delta_n \]
Which means that, there exist a set in $\mathbb{R}^n$ whose $P^n$ measure is at least  $1-\delta_n$ and for all sequences from this set it is true that
\[\left|\frac{\sum_{j=0}^n x_{j-k}x_{j-l}}{n}\right|< \gamma_n \]
Similarly there exist such sets for all $k\neq l$ in $\{0,1,\dots,l_n\}\times\{0,1,\dots,l_n\}$.
By a completely analogous procedure we can make sets which satisfy \ref{lem0-1} and \ref{lem0-2}.
The intersection of these sets $P$-measure at least $1-2\delta_n (l_n^2+1)$. As $\delta_n l_n^2 \rightarrow 0$, this proves the Lemma.
\end{proof}
The Lebesgue measure will be denoted by $\lambda$; its dimension is not specified, it will be always clear what it is.
\begin{Lem}\label{megj1}
If $n$ is large enough then the set $\irott A$ in Lemma \ref{Lem0} satisfies
\[2^{H(P)-2n\gamma_n} < \lambda(\irott A) < 2^{H(P)+n\gamma_n}\]
And for any $m$-dimensional continuous distribution $Q(\cdot)$ 
\[ \lambda(\irott T_Q)< 2^{H(Q)+n\gamma_n} \]
where $\irott T_Q$ is the set of typical sequences to $Q$, see (\ref{tipikussag})
\end{Lem}
\begin{proof}
 Since
\[1>P(\irott A)=\int_{\irott A } p(\Vx) \lambda(dx)>1-\delta\] 
by the previous Lemma, by using  $2^{-(H_{P_X}-n\gamma_n)}>p(\Vx)>2^{-(H_{P_X}+n\gamma_n)}$ on $\irott T_{P_X}$, and $\irott A \subset \irott T$, we get 
\[ 2^{H(P)+n\gamma_n)} > \lambda(\irott A) > (1-\delta)2^{H(P)-n\gamma_n)}>2^{H(P)-2n\gamma_n)}\]
if $n$ is large enough. Similarly from 
\[ 1>Q(\irott T_Q)=\int_{\irott T_Q} q(\Vx) \lambda(dx) \]
\[2^{H(Q)-n\gamma_n} > \lambda(\irott T_Q)\]
\end{proof}

The next lemma is an analogue of the Packing Lemma in \cite{Csiszar}
\begin{Lem}\label{Lem1}
For all $R>0, \delta>0,$, there exist at least $2^{n(R-\delta)}$ different sequences in $\mathbb R^n$ which are elements of the set $\irott{A}$ from  Lemma \ref{Lem0}, and for each pair of ISI channels with $h,\hat{h} \in \irott{H}_n$, $\sigma,\hat\sigma \in \irott{V}_n$, and for all $i \in \{1,2,\dots,M \}$ 

\begin{align}
 \lambda \left( \irott{T}_{P_{Y|X}^{(h,\sigma)}}(\Vx_i)\cap \bigcup_{j\neq i}
 \irott{T}_{P_{Y|X}^{(\hat{h},\hat{\sigma})}}(\Vx_j) \right) &\leq \notag \\
 &\hspace{-130pt}2^{-[n(|\I(\hat{h},\hat{\sigma})-R|_{+})-\Hh_{h,\sigma}(Y|X)]} \label{Pack-Lem}
\end{align}
provided that $n\geq n_0(n,m,\delta)$
\end{Lem}

\begin{proof}
We shall use the method of random selection. For fixed $n,m$ constans , let $\irott{C}_m$ be the family of all ordered collections $C=\{\Vx_1,\Vx_2,\dots ,\Vx_m \}$, of $m$ not necessarily different sequences in $\irott A$. Notice that if some $C=\{\Vx_1,\Vx_2,\dots ,\Vx_m \}\in \irott C_m$ satisfies (\ref{Pack-Lem}) for every $i$ and pair of Gaussian channels $(h,\sigma), (\hat h,\hat \sigma)$, then $x_i$'s are necessarily distinct.
For any collection $C \in \irott{C}_m$, denote the left-hand side of (\ref{Pack-Lem}) by $u_i(C,h,\hat{h})$. Since for $ \Vx \in T_{P}$ 
\[ \lambda\{\irott{T}_{P_{Y|X}^{h,\sigma}}(\Vx)\}\leq 2^{n(\Hh_{h,\sigma}(Y|X)-\gamma_n)}\] from Lemma \ref{megj1}, a $C \in \irott{C}_m$ satisfy (\ref{Pack-Lem}), if for all $i,h,\hat{h}$
\begin{align*}
&u_i (C)\circeq  \sum_{h,\hat{h}\in \irott{H}} u_i(C,h,\hat{h})\\
&\cdot 2^{n[\I(\hat{h},\hat{\sigma})-R]-\Hh_{h,\sigma}(Y|X)}  
\end{align*}
is at most 1, for every $i$.

Notice that, if $C\in \irott{C}_m$
\begin{equation}
 \frac{1}{m} \sum_{i=1}^m u_i(C)\leq 1/2 \label{pack-lem1}
\end{equation}
then $u_i(C)\leq 1$ for at least $\frac{m}{2}$ $i$ indices $i$. Further, if $C'$ is the subcollection, states the above indexes, then $u_i(C')\leq u_i(C)\leq 1$ for every such index $i$. Hence the Lemma will be proved, if for an $m$ with
\begin{align}
 2\cdot 2^{n(R-\delta)}\leq m \leq 2^{n(R-\frac{\delta}{2})}\label{lem1-1} 
\end{align}
we find a $C\in \irott{C}_m$ which  satisfy \ref{pack-lem1}.

Choose $C\in \irott{C}_m$ at random, according to uniform distribution from $\irott{A}$. In other words, let $W^m=(W_1,W_2,\dots ,W_m)$ be independent RV's, each uniformly distributed over $\irott{A}$. In order to prove that \ref{pack-lem1} is true for some  $C\in \irott{C}_m$, it suffices to show that

\begin{equation}
 \mathbb{E}u_i(W^m)\leq \frac{1}{2} \quad i=1,2,\dots,m \label{ujszam}
\end{equation}
 
To this end, we bound $\mathbb{E}u_i(W^m,h,\hat{h})$. Recalling that, $u_i(C,h,\hat{h})$ denotes the left-hand side of \ref{Pack-Lem}, we have
\begin{align}
& \mathbb{E}u_i(W^m,h,\hat{h})=\\
&\int\limits_{\Vy \in \mathbb{R}}
 \textnormal{Pr}\{\Vy\in \irott{T}_{P_{Y|X}^{h,\sigma}}(W_i)\cap \bigcup_{j\neq i} \irott{T}_{P_{Y|X}^{\hat{h},\hat{\sigma}}}(W_j)\}\label{lem1-2}
\end{align}
As the $W_j$ are independent identically distributed the probability the integration is bounded above by
\begin{align} &\sum_{j:j\neq i} \textnormal{Pr}\{\Vy\in \irott{T}_{P_{Y|X}^{h,\sigma}}(W_i)\cap \irott{T}_{P_{Y|X}^{\hat{h},\hat{\sigma}}}(W_j)\} = \\
&(m-1)\cdot \textnormal{Pr}\{\Vy\in \irott{T}_{P_{Y|X}^{h,\sigma}}(W_i)\}\cdot \textnormal{Pr}\{\Vy\in \irott{T}_{P_{Y|X}^{\hat{h},\hat{\sigma}}}(W_j)\}\label{lem1-3}
\end{align}
As the $W_j$'s are uniformly distributed over $\irott{A}$, we have for all fixed  $\Vy \in \mathbb{R}^n$
\[ \textnormal{Pr}\{\Vy\in \irott{T}_{P_{Y|X}^{h,\sigma}}(Z_i) \}= \frac{\lambda\{\Vx: \Vx\in \irott{T}_{P_X}, \Vy \in \irott{T}_{P_{Y|X}^{h,\sigma}}(\Vx) \}}{\lambda\{ \irott{A}\}}\]
The set in the enumerator is non-void only if $\Vy \in \irott{T}_{P_{Y}^{h,\sigma}}$. 
In this case it can be written as $\irott{T}_{\bar{P}_{X|Y}}(\Vy)$, where $\bar{P}$ is a conditional distribution, which
\[ P_X(\Va)P_{Y|X}^{h,\sigma}(\Vb|\Va)=P_Y^{h,\sigma}(\Vb)\bar{P}_{X|Y}(\Va|\Vb) \]
Thus by Lemma \ref{megj1}, and Lemma \ref{Lem0}
\begin{align*}
\textnormal{Pr}\{\Vy\in \irott{T}_{P_{Y|X}^{h,\sigma}}(Z_i) \} \leq \frac{2^{\Hh_{h,\sigma}(X|Y)+n\gamma_n}}{2^{\Hh(X)-2n\gamma_n)}}\\
=2^{-n(I(h,\sigma)-3\gamma_n})
\end{align*}
If $\Vy \in \irott{T}_{P_{Y}^{h,\sigma}}$, and $\textnormal{Pr}\{\Vy \in \irott{T}_{P_{Y|X}^{h,\sigma}}(W_i)\}=0$ otherwise. So, if we upper bound $\lambda(\irott{T}_{P_{Y}^{h,\sigma}})$ by $2^{\Hh_{h,\sigma}(Y)+n\gamma_n}$ - with the use of Lemma \ref{megj1} - from (\ref{lem1-3}), (\ref{lem1-2}) and (\ref{lem1-1}) we get,
\begin{align*}& \mathbb{E}u_i(W^m,h,\hat{h})\leq \\
&\lambda(\irott{T}_{P_{Y}^{h,\sigma}})(m-1)2^{-n[\I(h,\sigma)+\I(\hat{h},\hat{\sigma})-6\gamma_n]}\\
&\leq 2^{-n[\I(\hat{h},\hat{\sigma})-R+\delta-7\gamma_n]+\Hh_{h,\sigma}(Y|X)}
\end{align*}
Let $n$ be so large that $7\gamma_n < \delta/2$,
then we get
\[ \mathbb{E}u_i(W^m)\leq |\irott{H}_n^2| |\irott{V}_n^2| 2^{-n(\delta/2)}\]
which proves (\ref{ujszam})
\end{proof}
\begin{Lem}\label{lem3}
 For $\Vx \in \irott{A}$ from Lemma \ref{Lem0}, and $\Vy$ as is (\ref{modell}), and $\tilde{h}=\argmin_{h\in \irott{H}(n)}\|\Vy-h* \Vx\|$, and $\Vz=\Vy-\tilde{h}*\Vx$,
 \[\frac{\sum_{j=1}^n z_jx_{j-k}}{n} < \gamma_n \quad k\in \{0\dots l_n\}\]
\end{Lem}
\begin{proof}
 (Indirect) Suppose that
\[\frac{\sum_{j=1}^n z_jx_{j-k}}{n} = \lambda_k > \gamma_n \]
for some $k \in \{0\dots l_n\}$. Then let
\[ \hat{h}_j= \left\{ \sumfrac{\tilde{h}_j \, \textnormal{if}\, j\neq k }{\tilde{h}_j+\gamma_n \, \textnormal{if}\, j=k} \right. \]
We will show, that $\|\Vy-\hat{h}*\Vx_i \|< \|\Vy-\tilde{h}*\Vx_i\|$, which contradicts to the definition of $\tilde{h}$). 

Now,
\begin{align*}
 &\|\Vy-\hat{h}*\Vx_i\|^2=\sum_{j=1}^n (y_j-\sum_{g=0}^{l_n} \hat{h}_g x_{j-g})^2=\\
 &=\sum_{j=1}^n (y_j-\sum_{g=0}^{l_n} \tilde{h}_g x_{j-g}-\gamma_n x_{j-k})^2=\\
 &=\sum_{j=1}^n (z_j-\gamma_n x_{j-k})^2=\sum_{j=1}^n (z_j^2 - 2 \gamma_n z_j x_{j-k} + \gamma_n^2 x_{j-k}^2 )
\end{align*}
On account of (\ref{lem0-0}), 
\begin{align*}
 &\leq \|\Vz_i\|^2 -2 n\gamma_n \lambda_k + \gamma_n^2(n+\gamma_n)\leq \|\Vz_i\|^2 -(n-\gamma_n)\gamma_n^2 \\ &=\|\Vy-\tilde{h}*\Vx_i\|-n\gamma_n^2+\gamma_n^3<\|\Vy-\tilde{h}*\Vx_i\|
\end{align*}
\end{proof}
\begin{Lem}\label{lem4}
 Let $\delta>0$, and $\Vx \in \irott{A}$ from Lemma \ref{Lem0}. Let $h,\sigma\in \irott H\times \irott V$ and  $h^o,\sigma^o\in \irott H \times \irott V$ be two arbitrarily (ISI function, variance) pairs. Let $\Vy$ and $\Vx$ be such that 
\begin{align}
&\Vy \in \irott{T}_{Y|X}^{h,\sigma}(\Vx)\\ 
&h=\argmin_{h\in \irott{H}(n)}\|\Vy-h* \Vx\|\\
&\sigma^2=\|\Vy-h* \Vx\|/n
\end{align}
Then
\[ p_{Y|X}^{h^o,\sigma^o}(\Vy|\Vx) \leq 2^{-n [d((h,\sigma)\|(h^o,\sigma^o))-\delta]-H_{h,\sigma} (Y|X) } \]
Here \begin{math}
d((h,\sigma)\|(h^o,\sigma^o))=-\frac{1}{2}\log(\frac{\sigma^2}{\sigma_o^2})-\frac{1}{2}+\frac{\sigma^2+\|h-h^o\|^2}{2*\sigma_o^2} \label{szam} 
\end{math}
 is an information divergence for Gaussion distributions, positive if $(h,\sigma)\neq (h^o,\sigma^o)$.
\end{Lem}
\begin{proof}
 \begin{align} 
&P_{Y|X}^{h^o,\sigma^o}(\Vy|\Vx_i) = 2^{-n\left[-\frac{1}{n}\log\left(\frac{P_{Y|X}^{h^o,\sigma^o}(\Vy|\Vx_i)}{P_{Y|X}^{h,\sigma}(\Vy|\Vx_i)}\right)\right]+\log(P_{Y|X}^{h,\sigma}(\Vy|\Vx_i))]}\label{lem4-0} \end{align}

 and $\Vy \in \irott{T}_{Y|X}^{h,\sigma}(\Vx)$ by the definition, so: 
\begin{align}
-\log(P_{Y|X}^{h,\sigma}(\Vy|\Vx_i)) &\geq \Hh_{h,\sigma}(Y|X)-\gamma_n> \Hh_{h,\sigma}(Y|X)-\frac{\delta}{3} \label{lem4-3}                                   \end{align}
if $n$ is large enough.
 With this:
 \begin{align}
  &\log\left(\frac{P_{Y|X}^{h^o,\sigma^o}(\Vy|\Vx_i)}{P_{Y|X}^{h,\sigma}(\Vy|\Vx)}\right)=\log\left(\frac  {\frac{1}{(2\Pi)^\frac{n}{2}\sigma_o^n}exp(-\frac{\|\Vy-h*\Vx\|^2}{2\sigma_o^2})}
  {\frac{1}{(2\Pi)^\frac{n}{2}\sigma^n}exp(-\frac{\|\Vy-h*\Vx\|^2}{2\sigma^2})}\right)=\notag\\
  &n\log(\frac{\sigma}{\sigma_o})+\frac{\|\Vy-h*\Vx\|^2}{2\sigma^2}-\frac{\|\Vy-h^o*\Vx\|^2}{2\sigma_o^2}\leq\notag\\
  &\frac{n}{2}\log(\frac{\sigma^2}{\sigma_o^2})+\frac{n+n\gamma_n}{2}-\frac{\|\Vy-h^o*\Vx\|^2}{2*\sigma_o^2}\label{lem4-1}
 \end{align}
 Introduce the following notation $\Vz=\Vy-h*\Vx$, Then
 \begin{align*}
  &\|\Vy-h^o*\Vx\|^2=\|\Vz+h*\Vx -h^o*\Vx \|^2=\\
  &=\sum_{j=1}^n(z_j + \sum_{k=0}^{l_n} (h_k-h^o_k) x_{j-k})^2=\\
  &\sum_{j=1}^n(z_j^2+2z_j\sum_{k=0}^{l_n} (h_k-h^o_k) x_{j-k}+(\sum_{k=0}^{l_n} (h_k-h^o_k) x_{j-k})^2\\
  &=\|\Vz\|+2\sum_{k=0}^{l_n} (h_k-h^o_k)\sum_{i=1}^n z_jx_{j-k}+\\
  &+\sum_{j=1}^n(\sum_{k=0}^{l_n} (h_k-h^o_k) x_{j-k})^2=
 \end{align*} 
using Lemma \ref{lem3} 
\[\sum_{j=0}^n z_jx_{i,j-k}< n\gamma_n\]
 \begin{align*}
  &=\|\Vz\|+2\sum_{k=0}^{l_n} (h_k-h^o_k) n\gamma_n+\sum_{k=0}^{l_n} (h_k-h^o_k)^2 \sum_{j=0}^nx_{i,j-k}^2 +\\
  &+\sum_{j=0}^n\sum_{k\neq m}^{l_n}(h_m-h^o_m)(h_k-h^o_k)x_{i,j-k}x_{i,j-m}\geq \\
  &\geq \|\Vz\|-n4l_nP_n\gamma_n+(n-n\gamma_n)\|h-h^o\|^2-n(l_nP_n)^2\gamma_n
 \end{align*}
With this we can bound (\ref{lem4-1})
 \begin{align}
  &-\frac{1}{n}\log\left(\frac{P_{Y|X}^{h^o,\sigma^o}(\Vy|\Vx_i)}{P_{Y|X}^{h,\sigma}(\Vy|\Vx_i)}\right)=\notag\\
  &-\frac{1}{2}\log(\frac{\sigma^2}{\sigma_o^2})-\frac{1}{2}+\frac{\sigma^2+\|h-h^o\|^2}{2*\sigma_o^2}-\frac{\gamma_n}{2}-4l_nP_n\gamma_n\notag\\
  &-\gamma_nl4P_n^2-(l_nP_n)^2\gamma_n= \label{lem4-2}
 \end{align}
while $\|\Vz\|=n\sigma^2$. If $n$ large enough, then \[\max(4l_nP_n\gamma_n,4\gamma_n l_nP_n^2,(l_nP_n)^2\gamma_n)<\frac{\delta}{6}\] since $\lim_{n \to \infty}P_n^2 l_n^2 \gamma_n=0$. Using (\ref{szam}) we continue from (\ref{lem4-2})
 \begin{align*}
  &=d((h,\sigma)\|(h^o,\sigma^o))-\frac{\gamma_n}{2}\\
  &-(4l_nH_n\gamma_n+\gamma_nl4H_n^2+(l_nH_n)^2\gamma_n)\geq\\
  &\geq d((h,\sigma)\|(h^o,\sigma^o)) -\delta/2
 \end{align*}
Substituting this and (\ref{lem4-3}) to (\ref{lem4-0}) gives the desired result.
\end{proof}
Now we can state, and prove our main theorem
\begin{Tet}\label{mainthm}
 For arbitrarily given $R>0$ $\varepsilon>0$, and blocklength $n>n_0(R,\epsilon)$, there exist a code $(f,\phi)$ (coding/decoding function pair), with rate $\geq R-\varepsilon$
such that for all ISI channels, with parameters $h^o\in \mathbb R^{l_n}, |h_i^o|<P_n, \sigma^o<P_n, \sigma^o\neq 0$, the average error probability satisfies
\[ P_e(h^o,\sigma^o,f,\phi)\leq 2^{-n(E_r(R,h^o,\sigma^o)-\varepsilon)} \]
Here
\begin{align}
 &E_r(R,h^o,\sigma^o)\circeq \\
 &\min_{h,\sigma}\{d((h,\sigma)\|h^o,\sigma^o))+|\I(h,\sigma)-R|_+\} \label{sajatexp}
\end{align} 
where $d((h,\sigma)\|h^o,\sigma^o))$ is the information divergence (\ref{szam}), 
\end{Tet}
\begin{remark}\label{megj2}
The expression minimised above is a continuous function of $h^o, \sigma^o,  h, \sigma, R$
\end{remark}

\begin{proof}
 Let $\delta=\varepsilon/3$, and let 
\[C=\{\Vx_1,\Vx_2,\dots,\Vx_M\}\] 
the set of codesequences from Lemma \ref{Lem1}, so $M\geq 2^{n(R-\delta)}$. The coding function sends the $i$-th codeword for message $i$, $f(i)=\Vx_i$.

The decoding happens as follows:\
Let denote the ISI-type of $\Vy,\Vx_i$ by $h(i),\sigma(i)$ for all $i\in \{1,2,\dots,M\}$. Using these parameters we define the decoding rule as follows 
\[\phi(\Vy)=i \iff i=\argmax\limits_j I(h(j),\sigma(j))\]
in case of non-uniqueness, we declare an error.

Now we bound the error
 \begin{align}
  P_e=P_{out}+\frac{1}{M}\sum_{i=1}^{M}P^{h^o,\sigma^o}_{Y|X}(\phi(\Vy)\neq i|\Vx_i,e_1^c) \label{hibavsz}
 \end{align}
where $P_{out}$ denotes the probability of event $\irott E$ that the detected variance, for some $i \in \{1,2,\dots,M\}$ does not satisfy $|\sigma(i)^2-\frac{\|\Vy-h(i)*\Vx_i\|^2}{n}|< \gamma_n$. Bound the probability of this event. If $\irott E$ occurs then $\sigma(i)$ is extremal point of the approximating set of parameters, so $\\|\Vy-h(i)*\Vx_i\|^2>nP_n$. Since $h=(0,0,\dots,0)$ is element of the approximating set of ISI, this means that the power of the incoming sequence is greater than $nP_n$, the probability of this 
\[P_{out}\leq \left( 2 \int_{P_nl_n}^\infty \frac{1}{\sigma^n (2\pi)^{n/2}} e^\frac{-|z|^2}{2\sigma^2} \right)^n= \]
\[=(2\textnormal{erfc}(\frac{P_nl_n}{\sigma}))^n << e^{-n^{1+1/8}}\]
where $\textnormal{erfc}(\cdot)$ the complement normal error function. this probability converges to 0 faster that exponential, which - as we will see - means that in (\ref{hibavsz}) the second term is the dominant.

Consider the second term from (\ref{hibavsz}). If we sent $i$ then $\phi(\Vy)\neq i$ occurs if and only if 
\[\I(h(j),\sigma(j))\geq \I(h(i),\sigma(i))\]
 We know that
\[ \Vy \in  \irott{T}_{P_{Y|X}^{h(i),\sigma(i)}} (\Vx_i)\cap \irott{T}_{P_{Y|X}^{h(j),\sigma(j)}}(\Vx_j) \]
while we supposed that we are not in the event  $\irott E$.

So the probability of the second term of (\ref{hibavsz})
\begin{align*}
  &P_{Y|X}^{h^o,\sigma^o}(\phi(\Vy)\neq i|\Vx_i,e_1^c)= \\
   &\sum_{\sumfrac{(\hat{h},\hat\sigma),(h,\sigma) \in (\irott{H}_n,\irott{S}_n),(\irott{H}_n,\irott{S}_n)}{\I(\hat{h},\hat\sigma)\geq \I(h,\sigma)}}  \int\limits_{\sumfrac{\{ \irott{T}_{P_{Y|X}^{h(i),\sigma(i)}} (\Vx_i)\cap}{\irott{T}_{P_{Y|X}^{h(j),\sigma(j)(\Vx_j)\}}} }} P_{Y|X}^{h^o,\sigma^o}(\Vy|\Vx_i)\ud\Vy
 \end{align*}
With Lemma \ref{lem4} (substituting $(h,\sigma)=(h(j),\sigma(j))$, $(\hat{h},\hat\sigma)=(h(i),\sigma(i))$, $\delta=\varepsilon/3$) and from Lemma \ref{Lem1} we get
 \[P_e \leq \hspace{-25pt}\sum_{\sumfrac{(\hat{h},\hat\sigma),(h,\sigma) \in (\irott{H}_n,\irott{V}_n)^2}{\I(h(j),\sigma(j))\geq \I(h(i),\sigma(i))}} \hspace{-33pt}2^{-n\min\limits_{h,\sigma}[d((h(i),\sigma(i))\|(h^o,\sigma^o))+|\I(h(j),\sigma(j))-R|_{+}-2\varepsilon/3]} \]
if $n$ is large enough. From this - since the number of the approximating channel parameters grows subexponentially - we get
 \[P_e \leq 2^{-n\min_{h,\sigma}[d(\sigma^o\|\sigma)+|\I(h,\sigma)-R|_{+}-\varepsilon]} \]
\end{proof}

\section{Numerical Result} \label{numericalr}
We compare the new error exponent with Gallager's error exponent. Gallager derived the method to send digital information through channel with continuous time and alphabet, with given channel function. This result can be easily modified to discrete time, as in e.g. \cite{RCEEX,Shamai}.
The Linear Gaussian channel with discrete time parameter, with fading vector $h=(h_0,h_1,\dots,h_l)$can be  formulated as follows:
\begin{align}
\Vy=\uuline{H}\Vx + \Vz 
\end{align}
where $\Vx$ is the input vector and 
\begin{align}
 \uuline{H}=\left[ 
\begin{matrix}
             	h_0 & 0 & 0 & \dots &0 \\
		h_1 & h_0 & 0 & \dots &0 \\
		\vdots & \vdots & \ddots & \dots &\vdots \\
		h_l & h_{l-1} & h_{l-2} & \dots &0 \\
		0 & h_l & h_{l-1} & \dots &0 \\
		\vdots & \vdots & \vdots & \ddots &\vdots \\
		0 & 0 & 0 & \dots & h_l 
 \end{matrix} 
\right]
\end{align}
From \cite{Gallager} we get the idea to define a right and a left eigenbasis for  the matrix $H$. So the right eigenbasis $\Vr_1 ,\Vr_2 ,\dots ,\Vr_n$, is not else then the eigenbasis of $H^T H$ where  $^T$ means the transpose. And $\Vl_1,\Vl_2,\dots,\Vl_m$ is the left eigenbasis, where $\Vl_i = \frac{\uuline H \Vr_i}{|\uuline H \Vr_i|}$ (or basis of $H H^T$), and denoting $\lambda_i=|\uuline H \Vr_i|$. As in the work of Gallager $\Vr_i \Vr_j =\delta_{i,j}=\Vl_i^* \Vl_j$, because $\Vr_i$-s form an orthonormal eigenbasis and $\Vl_i \Vl_j = \frac{\Vr_i \uuline H^T  \uuline H \Vr_i}{|\Vr_i \uuline H^T  \uuline H \Vr_i|}=\Vr_i \Vr_j$.

So write $\Vx$ in a good basis $\Vr_1,\dots,\Vr_n$ we get $\Vx=\sum_{i=1}^n \tilde x_1 \Vr_i$ - we know that $\tilde x_i$-s form also an i.i.d. gaussian sequence. Write the output as $\Vy = \sum_{i=1}^n \tilde y_i \Vl_i$, we get
\begin{align}
 \tilde y_i= \lambda_i \tilde x_i +\tilde z_i
\end{align}
 for all $i$, where $\tilde z_i$ is the white Gaussian noise in the basis of $\Vl_1,\dots,\Vl_m$ (in which is also white). In many works this equation is used as a channel with fading, where $\lambda_i$s are i.i.d. random variables. This is a false approach. If the receiver knows the ISI $h$ then these constant can be computed. Is the ISI is a random vector from, e.g., i.i.d. random variables, then $\lambda_i$s are not necessarily i.i.d. 

If we interlace our codeword, with this formula we can get $n'=n+l$ parallel channels, each have SNR $\lambda_i/\sigma$.
We know that the error exponent given by Gallager for the $i$-th channel is
\[ E_r(\rho,\lambda_i)=-\frac{1}{n'}\ln\left[\int_y \left( \int_x q(x)p(y|x,\lambda_i)^{1/(1+\rho)} dx \right)^{1+\rho} \right]\]
If the input distribution $q(x)$ is the optimal, Gaussian distribution, with use of \cite{Shamai} the above expression can be rewritten as
\[E_r(\rho,\lambda_i)=-\frac{1}{n'}\sum_{i=1}^{n'}ln\left(1+\frac{\lambda_i}{\sigma(1+\rho)}\right)^{-\rho}\]
Now we can use the Szegő theorem from \cite{Toeplitz}, and we get that the average of the exponent, so the exponent of the system is
\[E_r(\rho)=\frac{1}{2\pi} \integ_{0}^{2\pi} -\ln \left[\left( 1+\frac{f(x)}{\sigma(1+\rho)}\right)^{-\rho}\right] dx \]
where $f(x)=\sum_{k=-\infty}^{\infty}(\sum_{j=0}^{l-|k|}h_{j}h_{j+|k|})e^{ik x}$, which is same as \cite{RCEEX,Shamai}

In the simulation we simulated a 4 dimensional fading vector whose components was randomly generated with uniform random distribution in $[0,1]$. For other randomly generated vectors, we get similar result. The two error exponent were positive in the same region, but for surprise the new error exponent was better (higher) than Gallager's one.



Figure \ref{2.abra} shows, that the new error exponent is always as good, or better than the Gallager's one.

\begin{figure}[ht]
\psfrag{Difference}{Difference}
\psfrag{SNR}{\parbox{10pt}{SNR (dB)}}
\psfrag{R}{R}
\psfrag{-10}{-10}
\psfrag{10}{10}
\psfrag{20}{20}
\psfrag{01}{0}
\psfrag{02}{0}
\psfrag{1}{1}
\psfrag{2}{2}
\psfrag{3}{3}
\psfrag{4}{4}
\psfrag{5}{5}
\psfrag{03}{0}
\psfrag{0.15}{0.15}
\psfrag{0.1}{0.1}
\psfrag{0.05}{0.05}
\includegraphics[width=8.75cm]{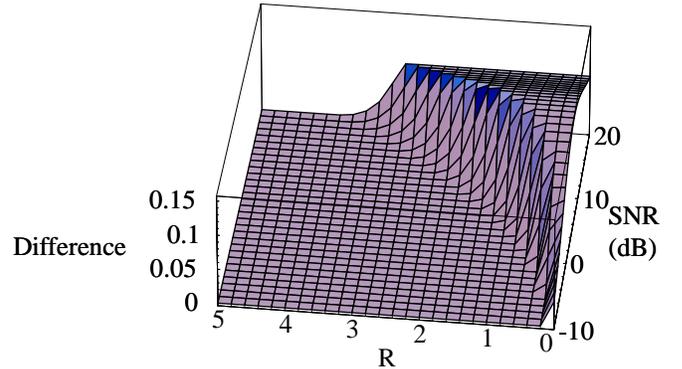}
\caption{Difference of the new error exponent and the Gallager's error exponent}\label{2.abra}
\end{figure}
The new method gives better error exponent, however it can be hardly computed. We could estimate the difference only in 4 dimension, because of the computational hardness to give the global optimum of a 4 dimensional function.

\section{Discussion} \label{Discussion}
Firstly our result can be used as a new lower bound to error exponents with no CSI at the transmitter. Note that, our error exponent (\ref{sajatexp}) is positive if the rate is smaller than the capacity.

Secondly it gives a new idea for decoding incoming signals without any CSI: Maybe it is worth to perform a more difficult maximisation, but not dealing with channel estimation. This can be done because of the universality of the code, which means, the detection method doesn't depend on the channel.

We have proved that if the ISI fullfills some criteria (see Theorem \ref{mainthm}), then the message can be detected, with exponentially small error probability. However these criteria can be relaxed, because in the Theorem $P_n\rightarrow \infty$ and $l_n\rightarrow \infty$, so any ISI with finite length and finite energy, and finite white noise variance, can be approximated well via the approximating set of parameters. 

It can be easily seen, that the lemmas and theorem remain true with small changes, if the input distribution is an arbitrarily chosen i.i.d. (absolute continuous or discrete) distribution. Only the  functional form of the mutual information $I(\cdot)$, and the entropy of the output variable $H_{(\cdot)}(Y)$ changes. So, this result can be used for lower bounding the error exponent, if non-gaussian i.i.d. random variables are used for the random selection of the codebook. However in these cases the entropy of the output can hardly be expressed in closed form.

With the result of Theorem \ref{mainthm}, we can define channel capacity for compound fading channels. If the fading remains unchanged during the transmission, and the fading length satisfies $l_n << n$, we can state the following theorem:

\begin{Tet}
Let $\irott F$ be an arbitrarily given not necessarily finite set of channel parameters, then the capacity of the ISI channel without any CSI, with channel parameter from $\irott F$, is 
\[C(\irott F)=\inf_{(h,\sigma)\in \irott F} I(h,\sigma)\] 
\end{Tet}
\begin{proof}
 In the limit of the set $\irott H_n$ is dense in the space of the real sequences with any length, so for every $(h,\sigma)\in \irott F$ there exists a sequence $(h_n,\sigma_n)\in \irott H_n \irott V_n$ such that $(h_n,\sigma_n)\rightarrow (h,\sigma)$. We know from remark \ref{megj2} that the error exponent in Theorem \ref{mainthm} is a continuous function, so the Theorem \ref{mainthm} proves that $C(\irott F)$ is an achievable rate. 

 Given linear gaussian channel with ISI $h$ and variance $\sigma$ the capacity is $I(h,\sigma)$ if the transmitter has no CSI.
\end{proof}

For some $(h,\sigma)$ our error exponent gives a better numerical result, than the random coding error exponent derived by Gallager \cite{Gallager}, improved by Kaplan and Shamai \cite{Shamai} (the random coding error exponent used here is deduced in Section \ref{numericalr}). 

This result is not so surprising, if we know that the Maximum Mutual Information (MMI) decoder gives better exponent in some cases (like in multiaccess environment) than the random coding error exponent derived by Gallager.

This work doesn't contradict with \cite{Ziv}. We know, that in the discrete case the MMI decoder is not else than the generalised likelihood (GML) decoder \cite{Ziv}, and also in \cite{Ziv} was showed, that GML decoder is not optimal in the non-memoryless case. However this is not the case in the continuous case, where the entropy of the incoming signal depends of the used parameter $(h,\sigma)$.

\end{document}